\documentclass[a4paper]{article}

\usepackage[utf8]{inputenc}
\usepackage{lmodern} 
\usepackage{tikz}
\usepackage[hyphens]{url}
\usepackage{hyperref}
\usepackage[hyphenbreaks]{breakurl}
\usepackage{xurl}
\hypersetup{breaklinks=true}
\usepackage{amsmath,amssymb,amsthm}
\usepackage{mathtools}
\usepackage[ruled,vlined,linesnumbered]{algorithm2e}
\usepackage{cleveref} 
\usepackage{pgfplots}
\usepackage{orcidlink}
\usepackage[shortlabels]{enumitem} 
\pgfplotsset{compat=1.16}
\crefname{algocf}{alg.}{algs.}
\Crefname{algocf}{Algorithm}{Algorithms}
\usepackage{booktabs} 

\DeclareMathOperator*{\leftext}{lext}
\DeclareMathOperator*{\rightext}{rext}
\DeclareMathOperator*{\leftextc}{left}
\DeclareMathOperator*{\rightextc}{right}

\newcommand{\No}[1]{}
\newcommand{\textefg}{indexable EFG}
\newcommand{\textefgs}{indexable EFGs}

\newcommand{\Textefgs}{Indexable EFGs}

\newtheorem{theorem}{Theorem}
\crefname{theorem}{theorem}{theorems}
\Crefname{theorem}{Theorem}{Theorems}

\newtheorem{lemma}{Lemma}
\crefname{lemma}{lemma}{lemmas}
\Crefname{lemma}{Lemma}{Lemmas}

\newtheorem{corollary}{Corollary}
\crefname{corollary}{corollary}{corollaries}
\Crefname{corollary}{Corollary}{Corollaries}

\newtheorem{definition}{Definition}
\crefname{definition}{definition}{definitions}
\Crefname{definition}{Definition}{Definitions}

\crefname{problem}{problem}{problems}
\Crefname{problem}{Problem}{Problems}

\newtheorem{observation}{Observation}
\crefname{observation}{observation}{observations}
\Crefname{observation}{Observation}{Observations}

\newtheorem{remark}{Remark}
\crefname{remark}{remark}{remarks}
\Crefname{remark}{Remark}{Remarks}

\crefname{assumption}{assumption}{assumptions}
\Crefname{assumption}{Assumption}{Assumptions}

\newcommand\sa{\mathsf{SA}\xspace}

\newcommand\mbwt{\mathsf{BWT}\xspace}

\newcommand\rmq{\textsf{RMQ}\xspace}

\title{Finding Maximal Exact Matches in Graphs}

\date{
Department of Computer Science, University of Helsinki, Finland
}

\author{
Nicola Rizzo\,\orcidlink{0000-0002-2035-6309}\\\texttt{nicola.rizzo@helsinki.fi} \and Manuel C\'aceres\,\orcidlink{0000-0003-0235-6951}\\\texttt{manuel.caceres@helsinki.fi} \and Veli M\"akinen\,\orcidlink{0000-0003-4454-1493}\\\texttt{veli.makinen@helsinki.fi}
}

\begin{document}

\maketitle

\begin{abstract}
We study the problem of finding maximal exact matches (MEMs) between a query string $Q$ and a labeled graph $G$. MEMs are an important class of seeds, often used in seed-chain-extend type of practical alignment methods because of their strong connections to classical metrics. A principled way to speed up chaining is to limit the number of MEMs by considering only MEMs of length at least $\kappa$ ($\kappa$-MEMs). However, on arbitrary input graphs, the problem of finding MEMs cannot be solved in truly sub-quadratic time under SETH (Equi et al., ICALP 2019) even on acyclic graphs.

In this paper we show an $O(n\cdot L \cdot d^{L-1} + m + M_{\kappa,L})$-time algorithm finding all $\kappa$-MEMs between $Q$ and $G$ spanning exactly $L$ nodes in $G$, where $n$ is the total length of node labels, $d$ is the maximum degree of a node in $G$, $m = |Q|$, and $M_{\kappa,L}$ is the number of output MEMs. We use this algorithm to develop a $\kappa$-MEM finding solution on indexable Elastic Founder Graphs (Equi et al., Algorithmica 2022) running in time $O(nH^2 + m + M_\kappa)$, where $H$ is the maximum number of nodes in a block, and $M_\kappa$ is the total number of $\kappa$-MEMs.

Our results generalize to the analysis of multiple query strings (MEMs between $G$ and any of the strings). Additionally, we provide some preliminary experimental results showing that the number of graph MEMs is an order of magnitude smaller than the number of string MEMs of the corresponding concatenated collection. 
\\[1ex]\textbf{Keywords} sequence to graph alignment, bidirectional BWT, r-index, suffix tree, founder graph
\\[1ex]\textbf{Acknowledgements} We thank an anonymous reviewer of WABI 2023 for pointing out an anomaly that lead us to correct the choice of a parameter value in our experiments.
\\[1ex]\textbf{Funding} This project received funding from the European Union’s Horizon 2020 research and innovation programme under the Marie Skłodowska-Curie grant agreement No.\ 956229, and from the Academy of Finland grants No.\ 352821 and 328877.
\end{abstract}

\section{Introduction}

Sequence alignment has been studied since the 1970s~\cite{needleman1970general,wagner1974string} and is a fundamental problem of computational molecular biology.
Solving the classical problems of \emph{longest common subsequence} (LCS) and \emph{edit distance} (ED) between two strings takes quadratic time with simple dynamic programs, and it was recently proven that no strongly subquadratic-time algorithms exist conditioned on the Strong Exponential Time Hypothesis (SETH)~\cite{bringmann2015quadratic,DBLP:journals/siamcomp/BackursI18}.
To overcome this hardness, researchers have used heuristics such as \emph{co-linear chaining}~\cite{DBLP:conf/spire/Abouelhoda07}: by taking (short) matches between the input strings, known as \emph{anchors}, we can take an ordered subset of these anchors and \emph{chain} them together into an alignment. Furthermore, when using \emph{maximal exact matches} (MEMs) as anchors, different co-linear chaining formulations capture both LCS~\cite{MS20} and ED~\cite{JGT22} in near-linear time. MEMs are also used in popular seed-and-extend read aligners~\cite{li2013aligning,marccais2018mummer4}. In fact, practical tools limit the number of MEMs by considering only $\kappa$-MEMs (MEMs of length at least $\kappa$) \cite{OhlebuschGK10,essaMEM}.

Extending alignment between sequences to sequence-to-graph alignment is an emerging and central challenge of \emph{computational pangenomics}~\cite{Maretal16}, as labeled graphs are a popular representation of pangenomes used in recent bioinformatics tools~\cite{Ma2022.01.07.475257,GC23,rautiainen2020graphaligner,li2020design}. We assume that a labeled graph $G = (V,E,\ell)$ ($\ell \colon V \to \Sigma^+$) is the reference pangenome of interest. Unfortunately, even finding exact occurrences of a given pattern in $G$ does not admit strongly subquadratic-time solutions under SETH~\cite{equi2023complexity}, and furthermore, a graph cannot be indexed in polynomial time to answer strongly subquadratic-time pattern matching queries~\cite{EMT21}. To circumvent this difficulty, research efforts have concentrated on finding parameterized solutions to (exact) pattern matching in labeled graphs~\cite{caceres2022parameterized,cotumaccio2021indexing,DBLP:conf/dcc/Cotumaccio22,rizzo2022solving}. Moreover, the use of MEMs and co-linear chaining has also been  extended to graphs~\cite{Ma2022.01.07.475257,GC23,rautiainen2020graphaligner,li2020design}.

In this paper, we study the problem of efficiently finding MEMs between a query string $Q$ and a labeled graph $G$, where we extend the MEM definition to capture any maximal match between $Q$ and the string spelled by some path of $G$. More precisely, our contributions are as follows:

\begin{itemize}
    \item In \Cref{sect:nodeMEMs}, we adapt the MEM finding algorithm between two strings of Belazzougui et al.~\cite{BCKM20} to find all $\kappa$-node-MEMs between $Q$ and $G = (V,E,\ell)$ in $O(m + n + M_\kappa)$ time, where $m = \lvert Q \rvert$, $n = \sum_{v \in V} \lvert \ell(v) \rvert$ is the cumulative length of the node labels, and $M_\kappa$ is the number of $\kappa$-node-MEMs (of length at least $\kappa$ and between the node labels and $Q$).
    \item Next, in \Cref{sect:3nodeMEMs}, we extend the previous algorithm to find all $\kappa$-MEMs spanning exactly $L$ nodes of $G$ in time $O(m + n \cdot L \cdot d^{L-1} + M_{\kappa, L})$, where $d$ is the maximum degree of any node $v \in V$ and $M_{\kappa, L}$ are the $\kappa$-MEMs of interest. Note that MEMs spanning less than $L$ nodes can occur multiple times in paths spanning exactly $L$ nodes, and our contribution is to introduce an efficient technique to filter out these MEMs.
    \item In \Cref{sect:efgMEMs}, we obtain the following results:
    \begin{itemize}
        \item We study $\kappa$-MEMs in indexable Elastic Founder Graphs (EFGs)~\cite{Equietal22}, a subclass of labeled acyclic graphs admitting a poly-time indexing scheme for linear-time pattern matching. Given an \textefg~$G$ of height $H$ (the maximum number of nodes in a graph block), we develop a suffix-tree-based solution to find all $\kappa$-MEMs spanning more than $3$ nodes in $G$ in $O(nH^2 + m + M_{\kappa,4+})$ time, where $M_{\kappa,4+}$ are the number of output MEMs. 
        \item Combined with the above results for $L=1,2,3$, we can find $\kappa$-MEMs of an \textefg~$G$ in $O(nH^2 + m + M_\kappa)$ time.
        \item We note that the previous results easily generalize to the batched query setting: by substituting $Q$ with the concatenation of different query strings $Q_1$, \dots, $Q_t$ of total length $m$, we compute all $\kappa$-MEMs between any query string and the graph with the same stated running time.
    \end{itemize}
    
    \item Finally, in \Cref{sect:experiments} we provide some preliminary experimental results on finding MEMs from a collection of strains of covid19. We use the bidirectional r-index~\cite{ANS22} as the underlying machinery. On the one hand, we build the r-index of the concatenation of the strains and find all $m_{\kappa}$ $\kappa$-MEMs. On the other hand, we build an \textefg\ of the strains and find an upper bound on all $M_{\kappa}$ $\kappa$-MEMs in this case. Our results indicate that $M_{\kappa}$ is an order of magnitude smaller than $m_{\kappa}$, thus confirming that graph MEMs compactly represent all MEMs.
\end{itemize}

\section{Preliminaries}\label{sect:preliminaries}

\paragraph*{Strings.}
We denote integer intervals by $[x..y]$, $x$ and $y$ inclusive. Let $\Sigma = [1..\sigma]$ be an alphabet.
A \emph{string} $T[1..n]$ is a sequence of symbols from $\Sigma$, that is, $T\in \Sigma^n$ where 
$\Sigma^n$ denotes the set of strings of length $n$ over $\Sigma$.
The \emph{length} of a string $T$ is denoted $|T|$ and the \emph{empty string} $\varepsilon$ is the string of length $0$. In this paper, we assume that $\sigma$ is always smaller or equal to the length of the strings we are working with. The concatenation of strings $T_1$ and $T_2$ is denoted as $T_1 \cdot T_2$, or just $T_1 T_2$. We denote by $T[x..y]$ the \emph{substring} of $T$ made of the concatenation of its characters from the $x$-th to the $y$-th, both inclusive; if $x = y$ then we also use $T[x]$ and if $y<x$ then $T[x..y] = \varepsilon$. The \emph{reverse} of a string $T[1..n]$, denoted by $\overline{T}$, is the string $T$ read from right to left, that is, $\overline{T} = T[n]T[n-1]..T[1]$. A \emph{suffix} (\emph{prefix}) of string $T[1..n]$ is $T[x..n]$ ($T[1..y]$) for $1\leq x \leq n$ ($1 \leq y \leq n$) and we say it is \emph{proper} if $x > 1$ ($y < n$). We denote by $\Sigma^*$ the set of finite strings over $\Sigma$, and also $\Sigma^+ = \Sigma^*\setminus \{\varepsilon\}$. String $Q$ \emph{occurs} in $T$ if $Q = T[x..y]$ for some interval $[x..y]$; in this case, we say that $[x..y]$ is a match of $Q$ in $T$.
Moreover, we study matches between substrings of $Q$ and $T$: a \emph{maximal exact match} (MEM) between $Q$ and $T$ is a triplet $(x_1,x_2,\ell)$ such that $Q[x_1..x_1 + \ell - 1] = T[x_2..x_2 + \ell - 1]$ and the match cannot be extended to the left nor to the right, that is, $x_1 = 1$ or $x_2 = 1$ or $Q[x_1-1] \neq T[x_2-1]$ (\emph{left-maximality}) and $x_1 + \ell = \lvert Q \rvert$ or $x_2 + \ell = \lvert T \rvert$ or $Q[x_1 + \ell] \neq T[x_2 + \ell]$ (\emph{right-maximality}). In this case, we say that the substring $Q[x_1..x_1 + \ell - 1]$ is a \emph{MEM string} between $Q$ and $T$. 
The \emph{lexicographic order} of two strings $T_1$ and $T_2$ is naturally defined by the total order $\le$ of the alphabet: $T_1 < T_2$ if and only if $T_1 \neq T_2$ and $T_1$ is a prefix of $T_2$ or there exists $y \geq 0$ such that $T_1[1..y]=T_2[1..y]$ and $T_1[y+1] < T_2[y+1]$.
We avoid the prefix-case by adding an \emph{end marker} $\$\not\in \Sigma$ to the strings and considering $\$$ to be the lexicographically smaller than any character in $\Sigma$.

\paragraph*{Labeled graphs.}

Let $G=(V,E,\ell)$ be a labeled graph with $V$ being the set of nodes, $E$ being the set of edges, and $\ell: V \to \Sigma^+$  being a function giving a label to each node. A length-$k$ \emph{path} $P$ from $v_1$ to $v_k$ is a sequence of nodes $v_1, \ldots, v_k$ connected by edges, that is, $(v_1,v_2),(v_2,v_3),\ldots,(v_{k-1},v_k) \in E$. A node $u$ \emph{reaches} a node $v$ if there is a path from $u$ to $v$. The label $\ell(P) \coloneqq \ell(v_1) \cdots \ell(v_k)$ of $P$ is the concatenation of the labels of the nodes in the path. For a node $v$ and a path $P$ we use $\lVert\cdot\rVert $ to denote its \emph{string length}, that is, $\lVert v\rVert  = |\ell(v)|$ and $\lVert P\rVert  = |\ell(P)|$. Let $Q$ be a query string. We say that $Q$ \emph{occurs} in $G$ if $Q$ occurs in $\ell(P)$ for some path $P$. In this case, the \emph{exact match} of $Q$ in $G$ can be identified by the triple $(i, P = v_1\ldots v_k, j)$, where $Q = \ell(v_1)[i..] \cdot \ell(v_2) \cdots \ell(v_{k-1}) \cdot \ell(v_k)[..j]$, with $1 \le i \le \lVert v_1\rVert $ and $1 \le j \le \lVert v_k\rVert $, and we call such triple a \emph{substring} of $G$.
Given a substring $(i,P,j)$ of $G$, we define its \emph{left-extension} $\leftext(i,P,j)$ as the singleton $\lbrace \ell(v_1)[i-1] \rbrace$ if $i > 1$ and otherwise as the set of characters $\lbrace \ell(u)[\lVert u\rVert ] \mid (u,v_1) \in E \rbrace$.
Symmetrically, the \emph{right-extension} $\rightext(i,P,j)$ is $\lbrace \ell(v_k)[j+1] \rbrace$ if $j < \lVert v_k\rVert $ and otherwise it is $\lbrace \ell(v)[1] \mid (v_k, v) \in E \rbrace$.
Note that the left (right) extension can be equal to the empty set $\emptyset$, if the start (end) node of $P$ does not have incoming (outgoing) edges. \Cref{fig:subpath} illustrates these concepts.

\begin{figure}
\centering
\begin{tikzpicture}
\node[fill=none,circle] (one-top) at (0,1) {\texttt{AGCTA}};
\node[fill=none,circle] (one-middle) at (0,0) {\texttt{CGCTC}};
\node[fill=none,circle] (one-bottom) at (0,-1) {\texttt{AGCAA}};
\node[fill=none,circle] (two) at (3,0) {\texttt{\underline{ACCGTA}}};
\node[fill=none,circle] (three-top) at (6,0.5) {\texttt{GT\underline{GGAA}}};
\node[fill=none,circle] (three-bottom) at (6,-0.5) {\texttt{GTGGAT}};
\node[fill=none,circle] (four) at (9,0) {\texttt{\underline{CC}AT}};
\draw[->] (one-top) to (two);
\draw[->] (one-middle) to (two);
\draw[->] (one-bottom) to (two);
\draw[->] (two) to (three-top);
\draw[->] (two) to (three-bottom);
\draw[->,line width=1.5pt] (three-top) to (four);
\draw[->] (three-bottom) to (four);
\end{tikzpicture}
\caption{Substring $\mathtt{ACCGTA}$ (underlined) with left-extension $\{\mathtt{A},\mathtt{C}\}$ and right-extension $\{\mathtt{G}\}$, and substring \texttt{GGAACC} (underlined, bold edge) with left-extension $\{\mathtt{T}\}$ and right-extension $\{\mathtt{A}\}$.
\label{fig:subpath}}
\end{figure}
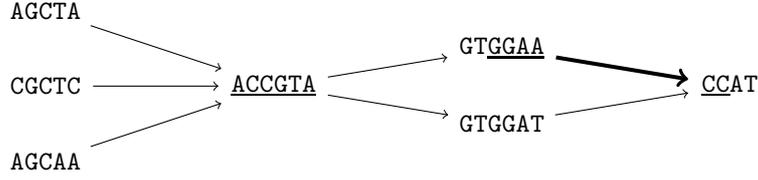

\paragraph*{Basic tools.}

A \emph{trie} or \emph{keyword tree}~\cite{Bri59} of a set of strings is an ordinal tree where the outgoing edges of each node are labeled by distinct symbols (the order of the children follows the order of the alphabet) and there is a unique root-to-leaf path spelling each string in the set; the shared part of two root-to-leaf paths spells the longest common prefix of the corresponding strings. In a \emph{compact trie}~\cite{DBLP:books/cu/Gusfield1997}, the maximal non-branching paths of a trie become edges labeled with the concatenation of labels on the path. The \emph{suffix tree} of $T \in \Sigma^*$ is the compact trie of all suffixes of the string $T' = T \$$. In this case, the edge labels are substrings of $T$ and can be represented in constant space as an interval of $T$. Such tree uses linear space and can be constructed in linear time, assuming that $\sigma \le \lvert T \rvert$, so that when reading the root-to-leaf paths from left to right, the suffixes are listed in their lexicographic order~\cite{DBLP:journals/algorithmica/Ukkonen95,Farach97}. As such, the order spelled by the leaves of the suffix tree form the \emph{suffix array} $\sa_T[1..|T'|]$, where $\sa_T[i]=j$ iff $T'[j..|T'|]$ is the $i$-th smallest suffix in lexicographic order.
When applied to a string $T$, the \emph{Burrows--Wheeler transform} (BWT) \cite{BW94} yields another string $\mbwt_T$ such that $\mbwt_T[i] = T'[\sa[i] - 1]$ (we assume $T'$ to be a circular string, i.e.\ $T'[-1] = T'[|T| + 1] = \$$).

Let $Q[1..m]$ be a query string. If $Q$ occurs in $T$, then the \emph{locus} or \emph{implicit node} of $Q$ in the suffix tree of $T$ is $(v,k)$ such that $Q = XY$, where $X$ is the path spelled from the root to the parent of $v$ and $Y$ is the prefix of length $k$ of the edge from the parent of $v$ to $v$. The leaves in the subtree rooted at $v$, also known as \emph{the leaves covered by $v$}, correspond to all the suffixes sharing the common prefix $Q$. Such leaves form an interval in the $\sa$ and equivalently in the BWT.
Let $aX$ and $X$ be the strings spelled from the root of the suffix tree to nodes $v$ and $w$, respectively. Then one can store a \emph{suffix link} from $v$ to $w$. Suffix links from implicit nodes are called implicit suffix links.

The \emph{bidirectional BWT}~\cite{DBLP:journals/iandc/SchnattingerOG12} is a compact BWT-based index capable of solving the MEM finding problem in linear time~\cite{BCKM20}. The algorithm simulates the traversal of the corresponding suffix trees to find maximal occurrences in both strings: in the first step, it locates the suffix tree nodes (intervals in the BWTs) corresponding to the maximal matches, that is, the MEM strings, and then it uses a cross-product algorithm to extract each MEM from the BWT intervals.

Let $B[1..n]$ be a bitvector, that is, a string over the alphabet $\{0,1\}$. There is a data structure that can be constructed in time $O(n)$ which answers $r = \mathtt{rank}(B,i)$ and $j=\mathtt{select}(B, r)$ in constant time, where the former operation returns the number of $1$s in $B[1..i]$ and the latter returns the position $j\le i$ of the $r$-th $1$ in $B$~\cite{clark1997compact,DBLP:conf/focs/Jacobson89}.
 
Let $D[1..n]$ be an array of integers. There is a \emph{range minimum query} data structure that can be constructed in $O(n)$ time which answers $\rmq_D(i,j)$ in constant time \cite{FH11}, where $\rmq_D(i,j)$ returns an index $k$ such that $D[k]$ is the minimum value in the subarray $D[i..j]$. We will use the following lemma that exploits range queries recursively.

\begin{lemma}
Let $D[1..n]$ be an array of integers. One can preprocess $D$ in $O(n)$ time so that given a threshold  $\Delta$, one can list all elements of $D$ such that $D[i]\leq \Delta$ in linear time in the size of the output.\label{lemma:mutustrick}
\end{lemma}
\begin{proof}
We can build the range minimum query data structure on $D$. Consider a recursive algorithm analogous to the one in \cite{Mut02} which starts with $k=\rmq_D(1,n)$. If $D[k] > \Delta$, the algorithm stops as no element in the range is at least $\Delta$. Otherwise, the algorithm reports $k$ and recursively continues with $\rmq_D(1,k-1)$ and $\rmq_D(k+1,n)$. Note that each recursive call performs exactly one $\rmq$ operation: if an element is reported the $\rmq$ is charged to this element, otherwise it is charged to its parent (in the recursion tree), and thus the number of $\rmq$ operations is linear in the output size.
\end{proof}

\section{Finding MEMs in Labeled Graphs}

Let us consider the problem of finding all \emph{maximal exact matches} (MEMs) between a labeled graph $G=(V,E,\ell)$, with $\ell: V \to \Sigma^+$, and a query string $Q \in \Sigma^+$.
\begin{definition}[MEM between a pattern and a graph]\label{def:mem}
Let $G = (V,E,\ell)$ be a labeled graph, with $\ell \colon V \to \Sigma^+$, let $Q \in \Sigma^+$ be a query string, and let $\kappa>0$ be a threshold.
Given a match $(i,P,j)$ of $Q[x..y]$ in $G$, we say that the pair $([x..y],(i,P,j))$ is \emph{left-maximal} (\emph{right-maximal}) if it \emph{cannot} be extended to the left (right, respectively) in both $Q$ and $G$, that is,
\begin{align*}
    (\mathsf{LeftMax}) & \quad x=1 \;\vee\; \leftext(i,P,j) = \emptyset \;\vee\; Q[x-1] \notin \leftext(i,P,j) \qquad\text{and}
     \\
    (\mathsf{RightMax}) & \quad y = \lvert Q \rvert \;\vee\; \rightext(i,P,j) = \emptyset \;\vee\; Q[y+1] \notin \rightext(i,P,j).
\end{align*}
We call $([x..y],(i,P,j))$ a $\kappa$-MEM iff  
    $\mathsf{LeftMax} \vee \lvert \leftext(i,P,j) \rvert\geq 2$,
    $\mathsf{RightMax} \vee \lvert \rightext(i,P,j) \rvert\geq 2$,
and $y-x+1\geq \kappa$, meaning that it is of length at least $\kappa$, it is left-maximal or its left (graph) extension is not a singleton, and it is right-maximal or its right (graph) extension is not a singleton.

\end{definition}

For example, with $Q=\mathtt{CACCGTAT}$, $\kappa=0$, $v$ being the first underlined node of \Cref{fig:subpath}, and $u$ being the second in-neighbor of $v$, then $([1..7],(5,uv,6))$ is a MEM since it is left and right maximal. Note that pair $([2..7],(1,v,6))$ is also a MEM since it is right-maximal, and the left extension of $(1,v,6)$ is not a singleton ($\leftext(v) = \lbrace \mathtt{A}, \mathtt{C} \rbrace$): this match is not left-maximal but our definition includes it as there are at least two different characters to the left.

We use this particular extension of MEMs to graphs---with the additional conditions on non-singletons $\leftext$ and $\rightext$---as it captures all MEMs between $Q$ and $\ell(P)$, where $P$ is a source-to-sink path in $G$.
Moreover, this MEM formulation (with $\kappa=1$) captures LCS through co-linear chaining, whereas avoiding the additional conditions would fail \cite{RCM23}.

The rest of this section is structured as follows. In \Cref{sect:nodeMEMs}, we show how to adapt the MEM finding algorithm of Belazzougui et~al.~\cite{BCKM20} for the case of node MEMs, which ignore the singleton conditions of \Cref{def:mem}. Then, in \Cref{sect:3nodeMEMs}, we show how to further generalize this approach to report all $\kappa$-MEMs spanning exactly $L$ nodes.

\subsection{MEMs in Node Labels\label{sect:nodeMEMs}}

We say that a match $([x..y], (i,v,j))$ is a \emph{node MEM} if it is left and right maximal w.r.t. $\ell(P)$ only in the string sense. In other words, a node MEM is a (string) MEM between $Q$ and $\ell(v)$ (especially in the case $x=1$ or $y=\ell(v)$).
For this, we consider the text
\[
    T_\text{nodes} = \prod_{v \in V} \mathbf{0} \cdot \ell(v),
\]
where $\mathbf{0}\notin \Sigma$ is used as a delimiter to prevent MEMs spanning more than a node label. 

Running the MEM finding algorithm of Belazzougui et~al.~\cite[Theorem 7.10]{BCKM20} on $Q$ and $T_\text{nodes}$ will retrieve exactly the node MEMs we are looking for. 
Given such a MEM $(x_1, x_2, \ell)$, to transform the coordinates of $T_\text{nodes}[x_2..x_2+\ell-1]$
into the corresponding graph substring $(i,P,j)$ we augment the index with a bitvector $B$ marking the locations of $\mathbf{0}$s of $T_\text{nodes}$, so that $r = \mathtt{rank}(B,x_2)$ identifies the corresponding node of $G$, $i =x_2-\mathtt{select}(B, r)$ and $j =i+\ell-1$.

\begin{lemma}\label{lemma:nodeMEMs}
    Given a labeled graph $G = (V,E,\ell)$, with $\ell: V \to \Sigma^+$, a query string $Q$, and a threshold $\kappa>0$, we can compute all node MEMs of length at least $\kappa$ between $Q$ and $G$ in time $O(n+m+N_\kappa)$, where $n$ is the total length of node labels, $m = |Q|$, and $N_\kappa$ is the number of output MEMs.
\end{lemma}

\begin{proof}
    The claim follows by applying the MEM finding algorithm of Belazzougui et~al.~\cite[Theorem 7.10]{BCKM20} on $T_\text{nodes}$ and $Q$, and using $B$ to extract the corresponding graph matches. We include a brief explanation of the original algorithm (following the later adaptation of~\cite[Algorithm 11.3]{MBCT15}) as we will modify this algorithm later. In our context, the algorithm takes as inputs the bidirectional BWT index on $T_\text{nodes}\$$ and the bidirectional BWT index on $Q\$$, and produces MEM strings $Q'$, with $|Q'|\geq \kappa$, and four BWT intervals $[i_{G}..j_{G}]$, $[i'_{G}..j'_{G}]$, $[i_Q..j_Q]$, and $[i'_Q..j'_Q]$, for each such $Q'$.
    The first two BWT intervals are such that $Q'$ occurs in $T_\text{nodes}$ $j_G-i_G+1=j'_G-i'_G+1$ times: the suffixes of $T_\text{nodes}$ having $Q'$ as a prefix have lexicographic ranks between $i_G$ and $j_G$, and the prefixes of $T_\text{nodes}$ having $\overline{Q'}$ (the reverse of $Q'$) as a suffix have co-lexicographic (lexicographic of the reverse) ranks between $i'_G$ and $j'_G$. Analogously, the last two BWT intervals are such that $Q'$ occurs in $Q$ $j_{Q}-i_{Q}+1=j'_{Q}-i'_{Q}+1$ times: the suffixes of $Q$ having $Q'$ as a prefix have lexicographic ranks between $i_{Q}$ and $j_{Q}$, and the prefixes of $Q$ having $\overline{Q'}$ as a suffix have co-lexicographic ranks between $i'_{Q}$ and $j'_{Q}$.
    For each of these four BWT intervals reported, the algorithm runs a \emph{cross-product} routine outputting all MEMs whose MEM string is $Q'$.
    Globally, the algorithm is linear in both the input and output size, since the exploration of MEM strings $Q'$ takes linear time in the size of the input strings, and since in total the cross-product routine runs in linear time in the number of BWT intervals considered as well as in the number of output MEMs~\cite[Theorem 7.10]{BCKM20}. 
    To see the linearity with respect to the input strings, let us study how the MEM strings $Q'$ are explored. The algorithm starts with intervals covering all suffixes/prefixes corresponding to a match of the empty string. It then executes recursively, so we can consider a generic step with intervals $[i_{G}..j_{G}]$, $[i'_{G}..j'_{G}]$, $[i_Q..j_Q]$, and $[i'_Q..j'_Q]$ corresponding to a match of the string $Q'$. Moreover, the algorithm maintains the invariant that the current match $Q'$ is right-maximal. In the recursive step, the algorithm first checks whether the match is left-maximal, in which case it reports the corresponding MEMs using the aforementioned cross-product algorithm. It then extends the match to the left with every possible character extension. For each such extension $aQ'$, it checks whether the extension contains a right-maximal match: if this is not the case the suffix tree of $T_\text{nodes}\#Q$ ($\# \not\in\Sigma$) does not have an internal node corresponding to $aQ'$, as all suffixes starting with $aQ'$ continue with the same symbol, otherwise (the extension contains a right-maximal match) the suffix tree has an internal node corresponding to $aQ'$. This exploration is bounded by the number of \emph{implicit Weiner links} in the suffix tree of $T_\text{nodes}\#Q$, which is linear in the input length~\cite[Observation 1]{BCKM20}. 
\end{proof}

\subsection{MEMs Spanning Exactly $L$ Nodes\label{sect:3nodeMEMs}}

Given a threshold $\kappa$, we want to find all $\kappa$-MEMs $([x..y],(i,P,j))$ spanning exactly $L$ nodes in $G$, that is, $\lvert P \rvert = L$.
Note that the MEMs obtained for $L=1$ are a subset of the ones obtained in \Cref{lemma:nodeMEMs}: for a node MEM $([x..y], (i, v, j))$ it might hold that $i = 1$ and $\lbrace Q[x-1] \rbrace = \leftext(1,v,j)$, or that $j = \lVert v\rVert$ and $\lbrace Q[y+1] \rbrace = \rightext(i,v,j)$.

As per \Cref{def:mem}, MEMs cannot be recognized without looking at the context of the paths in $G$ (sets $\leftext$ and $\rightext$).  
With this in mind, we consider the text
\begin{equation}\label{eq:tk}
    T_L \coloneqq
    \mathbf{0} \cdot \prod_{(u_1, \dots, u_L) \in \mathcal{P}_G^L} \Big(\leftextc(u_1) \cdot \ell(u_1) \cdots \ell(u_L) \cdot \rightextc(u_L) \cdot \mathbf{0} \Big),
\end{equation}
where $\leftextc(u)=c$ when $\leftext(u)=\{c\}$ and otherwise $\leftextc(u)=\#$, $\rightextc(u)=d$ when $\rightext(u)=\{d\}$ and otherwise $\rightextc(u)=\#$, $\mathbf{0}\neq \#$, $\mathbf{0},\# \notin \Sigma$, and 
\[
\mathcal{P}^L_G \coloneqq
    \bigg\lbrace
    P \biggm\vert \substack{\displaystyle P \;\text{path of}\; G, \\ \displaystyle \lvert P \rvert = L} 
    \bigg\rbrace.
\]

We have added the unique left- and right-extension symbols $c$ and $d$ to avoid reporting exact matches that can potentially be extended to longer paths. When these extensions are not unique (or empty), one can safely report a MEM, since there is a path diverting with a symbol different from that of the pattern (or the path cannot be extended further).
With these left- and right-extension symbols, it suffices to modify the MEM finding algorithm of \Cref{sect:nodeMEMs} to use some extra information regarding the starting position of each suffix inside string the $\ell(P)$, as explained next.

To avoid reporting MEMs spanning less than $L$ nodes (only if $L > 1$), we use an array $D[1..|T_L|]$ such that $D[k]=\infty$ if the $k$-th suffix $T_L[s..|T_L|]$ of $T_L$ in the lexicographic order is such that $T_L[s+1..|T_L|]$ is not starting inside node $u_1$ of a path $P=u_1 \cdots u_L$, otherwise $D[k]=|\ell(P)|-\ell(u_L)-i+2$, where suffix $T_L[s+1..|T_L|]$ starts at position $i$ inside $u_1$. That is, when $D[k]\neq\infty$, it tells the distance of the $k$-th suffix of $T_L$ in the lexicographic order to the start of the last node of the corresponding path. With the help of \Cref{lemma:mutustrick} on $D$, we can then adapt the MEM finding algorithm to output suffixes corresponding to MEMs spanning exactly $L$ nodes as follows.

\begin{lemma}\label{lem:mem_m_nodes}
    Let alphabet $\Sigma$ be of constant size.
    Given a labeled graph $G = (V,E,\ell)$, a pattern $Q \in \Sigma^m$, a threshold $\kappa \ge 1$, and an integer $L \ge 1$, we can compute an encoding of all MEMs of length at least $\kappa$ and spanning exactly $L$ nodes of $G$ in time $O(m + \lvert T_L \rvert + M_{\kappa,L})$.
    Here, $T_L$ is defined as in \Cref{eq:tk} and $M_{\kappa,L}$ is the number of output MEMs.
\end{lemma}

\begin{proof}
    We build the bidirectional BWT indexes for $T_L\$$ and $Q\$$, the suffix array of $T_L\$$, and preprocess $D[1..|T_L|]$ as in \Cref{lemma:mutustrick} in time $O(\lvert T_L \rvert+\lvert Q \rvert)$. We also preprocess, in linear time, a bitvector $B$ marking the locations of $\mathbf{0}$s of $T_L$ so that we can map in constant time a position $i$ in $T_L$ to the $r$-th path appended to $T_L$ for $r = \mathtt{rank}(B,i)$.

    The only modification of~\cite[Theorem 7.10]{BCKM20} required to only output MEMs spanning exactly $L$ nodes (and only if $L>1$) is to change its very last step when considering a MEM candidate $Q'$ with $|Q'|\geq \kappa$. Namely, the cross-product routine loops over all characters $a,b,c,d \in \Sigma \cup \{\#\}$ with $a\neq c$ and $b\neq d$, such that $aQ'b$ is a substring of $Q$ and $cQ'd$ is a substring of $T_L$. It then computes (in constant time) the intervals $[i_{aQ'b}..j_{aQ'b}]$, $[i'_{aQ'b}..j'_{aQ'b}]$, $[i_{cQ'd}..j_{cQ'd}]$, and $[i'_{cQ'd}..j'_{cQ'd}]$, where the first two are the intervals in the bidirectional BWT on $T_L$ corresponding to $aQ'b$ and the latter two are the intervals in the bidirectional BWT on $Q$ corresponding to $cQ'd$. After that, it outputs a triple $(k,k',|Q'|)$ representing each MEM, where $k\in [i_{aQ'b}..j_{aQ'b}]$ and $k' \in [i_{cQ'd}..j_{cQ'd}]$. It suffices to modify the first iteration using Lemma~\ref{lemma:mutustrick} to loop only over $k \in [i_{aQ'b}..j_{aQ'b}]$ such that $D[k]\leq |Q'|+1$. Our claims are that the running time stays linear in the input and output size on constant-size alphabet, and that only MEMs spanning exactly $L$ nodes are output. The latter claim follows directly on how array $D$ is defined and used with \Cref{lemma:mutustrick}. For the former claim, the cross-product part of the original algorithm is linear in the output size (also on non-constant-size alphabet) since for each combination of left- and right-extension considered, the work can be charged to the output. In our case, due to the use of \Cref{lemma:mutustrick}, some combinations may lead to empty outputs introducing an alphabet-factor (constant) multiplier on the input length.    
\end{proof}

\begin{remark}
Note that the algorithm in \Cref{lem:mem_m_nodes} works in time $O(m + n \cdot L \cdot d^{L-1} + M_{\kappa,L})$, where $n = \sum_{v \in V} \ell(v)$ is the total label length of $G$ and $d$ is the maximum in- or out-degree of a node.
Indeed, $T_{L}$ corresponds to the concatenation of length-$L$ paths of $G$: the number of paths containing label $\ell(v)$ (for a node $v$) is at most $L\cdot d^{L-1}$.
\end{remark}

\section{MEMs in Elastic Founder Graphs \label{sect:efgMEMs}}

The approach of Section~\ref{sect:3nodeMEMs} is exponential on $L$, so we can only use it for constant $L$ if aiming for a poly-time MEM finding routine.
For general labeled graphs, this may be the best achievable, as we cannot find MEMs with a threshold $\kappa$ between a pattern $Q$ and $G$ in truly sub-quadratic time unless the \emph{Orthogonal Vectors Hypothesis} ($\mathsf{OVH}$) is false~\cite{equi2023complexity},
and exponential-time indexing is required for truly sub-quadratic MEM finding with threshold $\kappa$ unless $\mathsf{OVH}$ is false~\cite{EMT21}: finding MEMs between $Q$ and $G$ finds matches of $Q$ in $G$, and if $\kappa = \lvert Q \rvert$ MEM finding is exactly equivalent to matching a pattern in a labeled graph.
To have a poly-time indexing of $G$ that can solve MEM finding in truly sub-quadratic time,
it is necessary to constrain the family of graphs in question.
Therefore, we now focus on indexable Elastic Founder Graphs (\textefgs), that are a subclass of labeled directed acyclic graphs (labeled DAGs) having the feature that they support poly-time indexing for linear-time queries~\cite{Equietal22}.
We will show that the same techniques used to query if $Q$ appears in \textefg\ $G$ can be extended to solve MEM finding on $G$ with arbitrary length threshold $\kappa$.

\begin{definition}[Elastic Founder Graph~\cite{Equietal22}]
    Consider a \emph{block graph} $G = (V,E,\ell)$, where $\ell \colon V \to \Sigma^+$, $V$ is partitioned into $k$ \emph{blocks} $V_1$, \dots, $V_k$, and edges $(u,v) \in E$ are such that $u \in V_i$, $v \in V_{i+1}$ for some $i \in [1..k-1]$.
    We say that a block graph is an \emph{indexable Elastic Founder Graph} (\textefg) if the \emph{semi-repeat-free property} holds: for each $v \in V_i$, $\ell(v)$ occurs in $G$ only as prefix of paths starting with some $w \in V_i$.
\end{definition}

Note that the semi-repeat-free property allows a node label to be prefix of other node labels in the same block, whereas it forbids them to appear as a proper suffix of other node labels nor anywhere else in the graph.
\Textefgs\ can be obtained from a set of aligned sequences, in a way such that the resulting \textefg\ spells the sequences but also their \emph{recombination}: for a gapless alignment, we can build in time linear to the size of the alignment an optimal \textefg\ with minimum height $H$ of a block, where the height of block $V_i$ is defined as $\lvert V_i \rvert$, solution generalized to the case with gaps by using an alternative height definition~\cite{DBLP:conf/cpm/0001M22}.

Let us now consider MEM finding with threshold $\kappa$ on an \textefg\ $G = (V,E,\ell)$.
We can use the general $\kappa$-MEM finding algorithm of \Cref{lem:mem_m_nodes} between $Q$ and $G$ spanning exactly $L$ nodes, with $L = 1,2,3$; then, we find $\kappa$-MEMs that span longer paths with a solution specific to \textefgs.
To find all MEMs between $Q$ and $G$ spanning more than three nodes, we index
\[
    T_3' \coloneqq \mathbf{0} \cdot \prod_{(u,v),(v,w) \in E} \Big( \ell(u)\ell(v)\ell(w)\mathbf{0} \Big)
    \qquad\text{where}\; \mathbf{0} \notin \Sigma.
\]
Equi et al.~\cite{Equietal22} showed that the suffix tree of $T_3'$ can be used to query string $Q$ in $G$, taking time $O(\lvert Q \rvert)$.
We now extend this algorithm to find MEMs between $Q$ and \textefg\ $G$ with threshold $\kappa$ and spanning more than $3$ nodes.
For simplicity, we describe a solution for case $\kappa = 1$ and later argue case $\kappa > 1$.

First, we augment the suffix tree of $T_3'$:
\begin{itemize}
    \item we mark all implicit or explicit nodes $\overline{p}$ such that the corresponding root-to-$\overline{p}$ path spells $\ell(u)\ell(v)$ for some $(u,v) \in E$, so that we can query in constant time if $\overline{p}$ is such a node;
    \item we compute pointers from each node $\overline{p}$ to an arbitrarily chosen leaf in the subtree rooted at $\overline{p}$;
    \item for each node $v \in V$ of the \textefg\ we build trie $T_v$ for the set of strings $\lbrace \overline{\ell(u)} : (u,v) \in E \rbrace$;
    \item for each leaf, we store the corresponding path $uvw$ and the starting position of the suffix inside $\ell(u)\ell(v)\ell(w)$.
\end{itemize}

\begin{observation}[{\cite[Lemma 9]{Equietal22}}]\label{obs:atomicedges}
Given an \textefg\ $G=(V,E,\ell)$, for each $(v,w) \in E$ string $\ell(v)\ell(w)$ occurs only as prefix of paths starting with $v$.
Thus, all occurrences of some string $S$ in $G$ spanning at least four nodes can be decomposed as $\alpha \ell(u_2) \cdots \ell(u_{L-1}) \beta$ such that: $(i)$ $u_2 \cdots u_{L-1}$ is a path in $G$ and $u_2$, \dots, $u_{L-1}$ are unequivocally identified;
$(ii)$ $\alpha = \ell(u_1)[i..\lVert u_1 \rVert]$ with $1 \le i \le \lVert u_1 \rVert$ for some $(u_1,u_2) \in E$;
and $(iii)$ $\beta = \ell(u_L)$ for some $(u_{L-1},u_L) \in E$ or $\beta = \ell(u_L) (\ell(u_{L+1})[1..j])$ with $1 \le j < \lVert u_{L+1} \rVert$ for some $(u_{L-1},u_L),(u_L,u_{L+1}) \in E$.
Note that $\alpha,\beta \neq \varepsilon$ and $\beta$ has as prefix a full node label, whereas $\alpha$ might spell any suffix of a node label.
\end{observation}

The strategy to compute long MEMs between $Q$ and $G$ is to first consider, with a left-to-right scan of $Q$, all MEMs $([x..y],(i,P,j))$ such that:
\begin{enumerate}[(i)]
    \item\label{cond:i} $\lvert P \rvert > 3$;
    \item\label{cond:ii} they satisfy conditions $\mathsf{LeftMax}$ and $\mathsf{RightMax}$ of \Cref{def:mem}; and
    \item\label{cond:iii} are maximal with respect to substring $Q[x..y]$, that is, there is no other MEM $([x'..y'],(i',P',j'))$ with $x \le x' \le y' \le y$.
\end{enumerate}
Next, we will describe how to modify our solution to compute also all the other MEMs spanning more than $3$ nodes.
Due to \Cref{obs:atomicedges}, if $\alpha \ell(u_2) \cdots \ell(u_{L-1}) \beta$ is a decomposition of $Q[x..y]$, all MEMs $([x'..y'],(i',P',j'))$ with $x \le x' < y' \le y$ spanning more than 3 nodes are constrained to involve some $u_i$ with $i \in [2..L-1]$.

Consider the following modification of \cite[Theorem 8]{Equietal22} that matches $Q[1..y]$ in $G$.
Let $\overline{p}$ be the suffix tree node of $T_3'$ reached from the root by spelling $Q[1..y]$ in the suffix tree until we cannot continue with $Q[y+1]$:
\begin{enumerate}
	\item\label{case:1} If we cannot continue with $\mathbf{0}$, $Q[1..y]$ is part of some MEM between $Q$ and $G$ spanning at most $3$ nodes, so we ignore it, take the suffix link of $\overline{p}$ and consider matching $Q[2..y]$ in $G$.
	\item\label{case:2} If we can continue with $\mathbf{0}$ and the occurrences of $Q[1..y]$ span at most two nodes in $G$, then we also take the suffix link of $\overline{p}$ and consider matching $Q[2..y]$.
    Thanks to the semi-repeat-free property, we can check this condition by retrieving any leaf in the subtree rooted at node $\overline{p}_\mathbf{0}$, reached by reading $\mathbf{0}$ from $\overline{p}$.
    	\item\label{case:3} In the remaining case, $Q[1..y] = \alpha \ell(u_2) \ell(u_3)$ for exactly one $u_2 \in V$, with $(u_2,u_3) \in E$, due to \Cref{obs:atomicedges}, and we follow the suffix link walk from $\overline{p}$ until we find the marked node $\overline{q}$ corresponding to $\ell(u_2) \ell(u_3)$: from $\overline{q}$ we try to match $Q[y+1..]$ until failure, matching $Q[y+1..y']$ and reaching node $\overline{r}$.
\end{enumerate}
By repeating the suffix walk and tentative match of case \ref{case:3} until we cannot read $\mathbf{0}$ from the failing node, we find the maximal prefix $Q[1..y]$ occurring in $G$ and its decomposition $\alpha \ell(u_2) \cdots \ell(u_{L-1}) \beta$ as per \Cref{obs:atomicedges}.
Indeed, we can find unique nodes $u_2$, \dots, $u_{L-1}$ by analyzing the (arbitrarily chosen) leaf of the subtree rooted at $\overline{q}$ in every iteration of case \ref{case:3}.
Moreover, we can retrieve:
\begin{itemize}
    \item set $U_1$ of pairs $(i,u)$ such that $(u,u_2) \in E$ and $\alpha = \ell(u)[i..\lVert u \lVert]$, by iterating over the leaves of $\overline{p}$;
    \item unique node $u_L$ such that $(u_{L-1},u_L) \in E$ and $\ell(u_L) = \beta$, if such $u_L$ exists; and
    \item set $E_L$ of triplets $(u,u',j)$ such that $(u_{L-1},u),(u,u') \in E$ and \\ $\ell(u)\ell(u')[1..j] = \beta$.
\end{itemize}
Then $([1..y],(i, u_1 \cdot u_2 \cdots u_{L-1} \cdot u \cdot u_{L+1} ,j))$ is a MEM between $Q$ and $G$ for all $(i,u_1) \in U_1$ and $(u,u_{L+1},j) \in E_L$, and also $([1..y],(i, u_1 \cdot u_2 \cdots u_{L-1} \cdot u_L, \lVert u_L \rVert))$ is a MEM for all $(i,u_1) \in U_1$, if $u_L$ exists: these MEMs satisfy conditions \ref{cond:i}, \ref{cond:ii}, and \ref{cond:iii}, and $U_1$, $u_2 \cdots u_{L-1}$, $u_L$, and $E_L$ form a compact representation of all MEMs spelling $Q[1..y]$.

So far the procedure computes all MEMs spanning more than 3 nodes, satisfying $\mathsf{LeftMax}$ and $\mathsf{RightMax}$, and spelling maximal $Q[1..y]$.
We can extend it to find all MEMs satisfying the first two constraints and spelling any substring $Q[x..y]$, with $Q[x..y]$ maximal.
Let $\hat{x}$ be the index for which we have computed MEMs spelling $Q[\hat{x}..y]$ ($\hat{x} = 1$ in the first iteration).
If cases \ref{case:1} or \ref{case:2} hold, we can start to search MEMs spelling $Q[\hat{x}+1..]$ in amortized linear time, since we follow the suffix link of $\overline{p}$.
If case \ref{case:3} holds, we can restart the algorithm looking for MEMs spelling $Q[\hat{x}'..y]$, where $\hat{x}' = \hat{x} + \lvert \alpha \ell(u_2) \cdots \ell(u_{L-2}) \rvert$.
We are not missing any MEM satisfying conditions \ref{cond:i}, \ref{cond:ii}, and \ref{cond:iii}:
due to the semi-repeat-free property, 
any MEM $([x..y],(i',P',j'))$ with $\hat{x} < x < \hat{x}'$ spanning more than $3$ nodes shares substring $\ell(u_k)\ell(u_{k+1})$ with the previously computed MEM, for some $k \in [2..L-3]$, and is such that $\hat{x}' < y$ since we assume \ref{cond:iii} to hold; the algorithm would have matched $Q[\hat{x}..y]$ with case \ref{case:3} in the previous iteration, leading to a contradiction.
The time globally spent reading $Q$ is still $O(\lvert Q \rvert)$, because each character of $Q$ is considered at most twice.

Finally, we are ready to describe how to compute all remaining MEMs $([x..y],(i,P,j))$ between $Q$ and \textefg\ $G$ spanning at least 4 nodes, that is, MEMs such that condition \ref{cond:i} holds and at least one of \ref{cond:ii} and \ref{cond:iii} do not: it is easy to see that $Q[x..y]$ must be contained in the MEMs that we have already computed; also, since they span at least 4 nodes their matches must involve some of nodes $u_2$, \dots, $u_{L-1}$ of MEMs satisfying \ref{cond:i}, \ref{cond:ii}, and \ref{cond:iii}.
Indeed, whenever case \ref{case:3} holds and we decompose $Q[\hat{x}..y]$ as $\alpha \ell(u_2) \cdots \ell(u_{L-1}) \beta$, we can find set $U_{\mathrm{RT}}$ of pairs $(v,j)$, with $v \in V$ and $1 \le j \le \lVert v \rVert$, such that $(v,j) \in U_{\mathrm{RT}}$ iff $(i, P = u \cdot u_2 \cdots u_{b-1} \cdot v, j)$ is a match of $Q[\hat{x}..y']$ in $G$, with $(i,u) \in U_1$, $\lvert P \rvert < L$, $y' < y$, and $Q[y'+1] \notin \rightext(1,u,j)$---verifying $\mathsf{RightMax}$ and describing a MEM where \ref{cond:iii} fails---or $\lvert \rightext(1,u,j) \rvert \ge 2$---verifying the non-singleton condition of \Cref{def:mem} and describing a MEM where \ref{cond:ii} fails.
We can gather all the elements of $U_{\mathrm{RT}}$ during each descending walk in the suffix tree of $T_3'$, since they correspond to the leaves of subtrees of branching nodes in the tentative match of $Q[\hat{x}..y]$.
Analogously, we can find set $U_{\mathrm{LT}}$ of pairs $(i,v)$, with $v \in V$ and $1 \le i \le \lVert v \rVert$, such that $(i,v) \in U_{\mathrm{LT}}$ iff $(i,v) = (1,u_i)$ for $2 \le i \le L - 1$ and $\leftext(u_i) \ge 2$, or $(i, P = v \cdot u_b \cdots u_{L-1}, \lVert u_{L-1} \rVert)$ is a match of $Q[x..y - \lvert \beta \rvert]$ in $G$, with $x > \hat{x}$ and $Q[x-1] \notin \leftext(i,v,\lVert v \rVert)$.
We can compute $U_{\mathrm{LT}}$ by analyzing the leaves of subtrees of branching nodes in the walk in $T_{u_i}$ spelling $\overline{\ell(u_i)}$, with $2 \le i \le L-1$.
Sets $U_1$, $u_2\cdots u_{L-1}$, $u_L$, $E_L$, $U_{\mathrm{LT}}$ and $U_{\mathrm{RT}}$ are a compact representation of all MEMs spanning at least $4$ nodes and involving (any substring of) $Q[\hat{x}..y]$: a cross-product-like algorithm that matches elements of $U_1$ or $L$ with elements of $u_L$, $E_L$, or $U_{\mathrm{RT}}$, joined by the relevant part of $u_2 \cdots u_{L-1}$, can explicitly output the MEMs spanning more than 3 nodes in linear time with respect to the size of the output, by exploiting the fact that $U_{\mathrm{LT}}$ and $U_{\mathrm{RT}}$ are computed and ordered block by block.

\begin{theorem}
Let alphabet $\Sigma$ be of constant size, and let $G=(V,E,\ell)$ be an indexable Elastic Founder Graph of height $H$, that is, the maximum number of nodes in a block of $G$ is $H$.
An encoding of MEMs between query string $Q\in \Sigma^m$ and $G$ with arbitrary length threshold $\kappa$ can be reported in time $O(n H^2 + m +M_\kappa)$, where $n=\sum_{v\in V} \lVert v \rVert$ and $M_\kappa$ is the number of MEMs of interest.
\label{thm:efg-mems}
\end{theorem}
\begin{proof}
We can apply the algorithm of \Cref{lem:mem_m_nodes} to find $\kappa$-MEMs spanning $L$ nodes, with $L = 1,2,3$, taking time
$O(\lvert Q \rvert + \lvert T_1 \rvert + \lvert T_2 \rvert + \lvert T_3 \rvert + M_{\kappa,1} + M_{\kappa,2} + M_{\kappa,3})$.

Let $M_{\kappa,4+}$ be the number of MEMs satisfying threshold $\kappa$ and spanning at least $4$ nodes in $G$.
The suffix tree of $T_3'$ can be constructed in time $O(\lvert T_3' \rvert)$ and the described modification of a descending suffix walk on $Q$ takes constant amortized time per step, assuming constant-size alphabet.
The time spent gathering $U_1$, $u_2 \cdots u_{L-1}$, $E_L$, $U_\mathrm{LT}$, and $U_{\mathrm{RT}}$, forming an encoding of the MEMs involving $Q[\hat{x}..y]$, can be charged to $M_{\kappa,4+}$ because each element of $U_1$, $E_L$, $U_\mathrm{LT}$, and $U_{\mathrm{RT}}$ corresponds to one or more MEMs, that could be retrieved in an explicit form with a cross-product-like procedure.
Indeed: for $U_1$ we can retrieve all leaves of the subtree rooted at $\overline{p}$ of the suffix tree of $T_3'$; for $E_L$ and $U_{\mathrm{RT}}$, we can do the same for node $\overline{r}$ reached by the last tentative match of $Q[y+1..]$, and for branching nodes reached during every tentative match; for $U_\mathrm{LT}$, using a compact trie and \emph{blind search}~\cite{DBLP:journals/jacm/FerraginaG99} in the representation of each $T_u$ allows to compare only the branching symbols.
Finally, it is easy to see that in case \ref{case:3}, after we decomposed $\lvert Q[\hat{x}..y] \rvert$ as $\alpha \ell(u_2) \cdots \ell(u_{L-1}) \beta$ as in \Cref{obs:atomicedges}, we know the length of strings $\alpha$, $\ell(u_2)$, \dots, $\ell(u_{L-1})$, and $\beta$, so we can postpone the computation of sets $U_1$, $E_L$, $U_\mathrm{LT}$, and $U_\mathrm{RT}$ and avoid computing MEMs of length smaller than $\kappa$.
Thus, finding an encoding of all MEMs between $Q$ and $G$ with threshold $\kappa$ and spanning more than $3$ nodes takes $O(\lvert Q \rvert + \lvert T_3' \rvert + M_{\kappa,4+})$ time.

The stated time complexity is reached due to the fact that $\lvert T_3 \rvert$ dominates $\lvert T_3' \rvert$, $\lvert T_2 \rvert$, and $\lvert T_1 \rvert$, and for \textefgs\ $\lvert T_3 \rvert \in O(n H^2)$, since every character of every node label $\ell(u)$ gets repeated at most $H^2$ times, which is an upper bound on the number of paths of length $3$ containing $u$.
\end{proof}

\begin{corollary}
The results of Lemmas~\ref{lemma:nodeMEMs}~and~\ref{lem:mem_m_nodes} and Theorem~\ref{thm:efg-mems} hold when query $Q[1..m]$ is replaced by a set of queries of total length $m$. The respective algorithms can be modified to report MEMs between the graph and each query separately.
\label{cor:setofqueries}
\end{corollary}
\begin{proof}
Consider a concatenation $Q=Q^1\$ Q^2\$ \cdots Q^d$ of $d$ query sequences, where $\$$ is a unique symbol not occurring in the queries nor in the graph. 
No MEM can span over such unique separator and hence the MEMs between graph $G$ and concatenation $Q$ are the same as those between $G$ and each $Q^i$. It is thus sufficient to feed concatenation $Q$ as input to the algorithms and project each resulting MEM to the corresponding query sequence.  
\end{proof}

\begin{corollary}
The algorithms of Lemmas~\ref{lemma:nodeMEMs}~and~\ref{lem:mem_m_nodes}, Theorem~\ref{thm:efg-mems}, and Corollary~\ref{cor:setofqueries} can be modified to report only MEMs that occur in text $T$ formed by concatenating the rows (ignoring gaps and adding separator symbols) of the input MSA of the \textefg. This can be done in additional $O(|T|+r\log r)$ time and $O(r\log n)$ bits of space, and with multiplicative factor $O(\log \log n)$ added to the running times of the respective algorithms, where $r$ is the number of equal-letter runs in the BWT of $T$.
\label{cor:adding-r-index}
\end{corollary}
\begin{proof}
Lemma~\ref{lemma:nodeMEMs} does not need any modification as the node labels are automatically substrings of $T$. Same applies for the edge labels, but for longer paths we need to make sure we do not create combinations not supported by $T$. This can be accomplished with the help of the $r$-index: with the claimed time and space one can build the run-length encoded BWT of $T$ \cite{NKT22} and the associated data structures to form the counting version of the $r$-index that supports backward step in $O(\log \log n)$ time \cite{GNP20}. As we concatenate paths consisting of $L$ nodes for MEM finding in Lemma~\ref{lem:mem_m_nodes}, we can first search them using the $r$-index, and only include them if they occur in $T$. MEMs spanning more than $3$ nodes in Theorem~\ref{thm:efg-mems} and Corollary~\ref{cor:setofqueries} can be searched afterwards with the $r$-index to filter out those MEMs not occurring in $T$; these MEMs cannot mutually overlap each other in $Q$ by more than one full node label, so the running time of the verification can be charged on the size of the elastic founder graph.
\end{proof}

\section{Experiments}\label{sect:experiments}

The benefit of Corollary~\ref{cor:adding-r-index} over the mere use of $r$-index for MEM finding \cite{Rosetal22} is that a MEM can occur many times in a repetitive collection while the occurrences starting at the same column of a MSA of the collection can be represented by a small number of paths in the indexable Elastic Founder Graph. 

To test this hypothesis, we implemented the MEM finding algorithm using the bidirectional r-index~\cite{ANS22}. The code is available at \url{https://github.com/algbio/br-index-mems}. The implementation works both for MEM finding between two sequences and between a sequence and a graph. In the case of a graph, the implementation covers the algorithms of Sections~\ref{sect:nodeMEMs} and \ref{sect:3nodeMEMs} up to paths of length $3$ nodes. 
If the input graph is an \textefg\ and $\kappa$ is at most the length of the shortest string spelled by an edge, the implementation outputs all $\kappa$-MEMs, yet some longer MEMs are not fully extended and/or may be reported in several pieces.\footnote{In an earlier version of this paper, we mistakenly set $\kappa$ as the minimum length string spelled by paths spanning 3 nodes, but we later observed that this may lead the implementation to miss some $\kappa$-MEMs.} With these considerations, the implementation gives an upper bound on the number of $\kappa$-MEMs in the case of graphs, but it provides the exact number of MEMs in the case of two sequences as input.

We performed experiments with the same multiple sequence alignment (MSA) of covid19 strains as in~\cite{MCENT20}. We first filtered out strains whose alignments had a run of gaps of length more than 100 bases.\footnote{We used this filter to address a current limitation of the MSA segmentation algorithms for \textefg, where a long gap at the beginning or end of an MSA row hurts the effectiveness of the method.} Then we extracted a sub-MSA of 100 random strains from the remaining and extracted MSAs of the first 20, 40, 60, and 80 strains from this MSA of 100 strains. For each such dataset, we built the bidirectional r-index of the sequences (without gaps) and the \textefg\ of the MSA. The latter were constructed using the tool \url{https://github.com/algbio/founderblockgraphs} with parameters \verb+--elastic --gfa+.

We used $\kappa=12$ in all experiments, as this was the length of the shortest string spelled by an edge over all \textefgs. For the queries, we extracted 1000 substrings of length 100 from the first 20 strains. For each query, we selected two random positions and mutated them with equal probability for A, C, G, or T. The queries were then concatenated into a long sequence and the bidirectional r-index was built on it as described by \Cref{cor:setofqueries}. The MEMs were computed between the queries and the respective text/graph index.

The number of MEMs for each index is reported in Figure~\ref{fig:numberofMEMs} and the number of runs in the two Burrows--Wheeler transforms of each index is reported in Figure~\ref{fig:numberofruns}.
As can be seen from the results, the number of MEMs is greatly reduced when indexing the graph compared to indexing the collection of strains. In this setting, almost all graph MEMs have a counterpart in the collection: we implemented also a filtering mode analogous to Corollary~\ref{cor:adding-r-index} so that only graph MEMs occurring in the original strains are reported. The number of path MEMs (and thus total number of MEMs) reported dropped by 92, 97, 90, 91, and 91 for the case of \textefg\ on 20, 40, 60, 80, and 100 strains, respectively. That is, less than 2\% of the reported graph MEMs were recombinants of the input strains. However, since the implementation outputs long MEMs in pieces, this percentage may be higher for full-length graphs MEMs.
\begin{figure}[p]
\centering
\begin{tikzpicture}
    \begin{axis}[              xlabel={Number of strains},        ylabel={Number of MEMs},        xmin=20, xmax=100,              ymin=10, ymax=1000000,              legend pos=outer north east,              xmajorgrids=true,              ymajorgrids=true,              grid style=dashed,            ymode=log, ]
    \addplot[color=blue, mark=square] coordinates {
      (20, 47556)
      (40, 75877)
      (60, 113727)
      (80, 151672)
      (100, 189522)
    };
    \addlegendentry{sequences}
    \addplot[color=red, mark=*] coordinates {
      (20, 230)
      (40, 59)
      (60, 139)
      (80, 187)
      (100, 319)
    };
    \addlegendentry{efg-nodes}
    \addplot[color=green, mark=x] coordinates {
      (20, 460)
      (40, 43)
      (60, 64)
      (80, 78)
      (100, 120)
    };
    \addlegendentry{efg-edges}
    \addplot[color=black, mark=triangle] coordinates {
      (20, 4451)
      (40, 6001)
      (60, 6843)
      (80, 7654)
      (100, 9245)
    };
    \addlegendentry{efg-paths}
    \addplot[color=purple, mark=triangle*] coordinates {
      (20, 5144)
      (40, 6106)
      (60, 7049)
      (80, 7922)
      (100, 9687)
    };
    \addlegendentry{efg-total}
    \end{axis}
  \end{tikzpicture}
  \caption{Number of MEMs with different indexes and varying number of covid19 strains. Here sequences refers to bidirectional r-index. For \textefg\ the results are shown for nodes, edges, and paths of length 3 nodes (these also include longer MEMs counted multiple times). Line efg-total is the total number of EFG MEMs. Note the logarithmic scale on the $y$-axis. \label{fig:numberofMEMs}}
\end{figure}
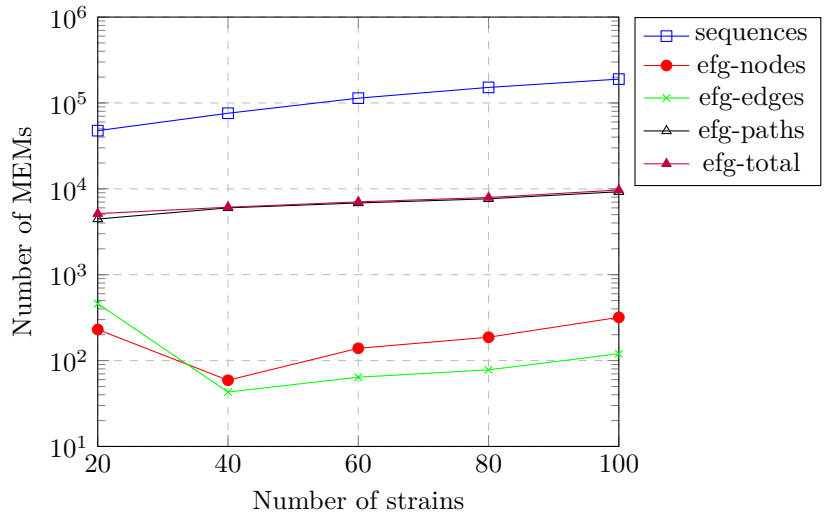  
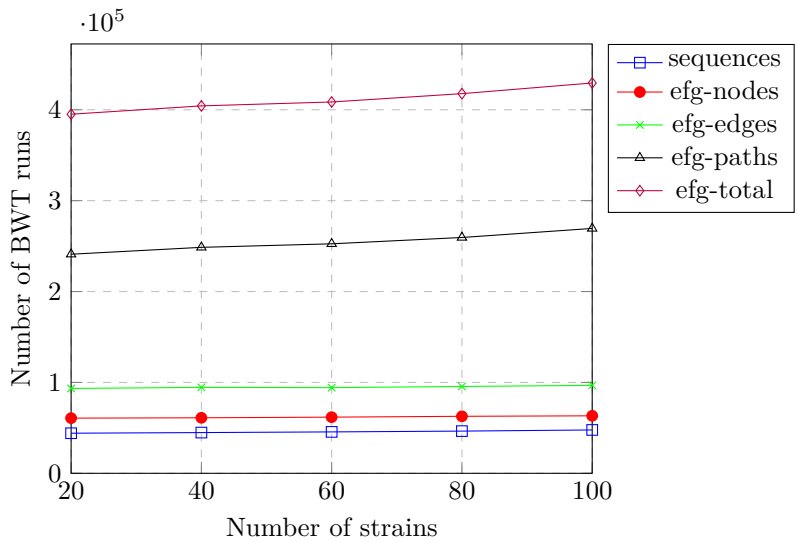
\begin{figure}[p]
  \centering
  \begin{tikzpicture}
    \begin{axis}[      xlabel={Number of strains},      ylabel={Number of BWT runs},      xmin=20, xmax=100,      ymin=0,      legend pos=outer north east,      xmajorgrids=true,      ymajorgrids=true,      grid style=dashed,    ]
    \addplot[color=blue, mark=square] coordinates {
      (20, 44113)
      (40, 44818)
      (60, 45560)
      (80, 46434)
      (100, 47704)
    };
    \addlegendentry{sequences}
   \addplot[color=red, mark=*] coordinates {
       (20, 60696)
       (40, 61090)
       (60, 61796)
       (80, 62722)
       (100, 63260)
    }; 
    \addlegendentry{efg-nodes}
    \addplot[color=green, mark=x] coordinates {
       (20, 93358)
       (40, 94558)
       (60, 94260)
       (80, 95539)
       (100, 96787)
   };
   \addlegendentry{efg-edges}
   \addplot[color=black, mark=triangle] coordinates {
       (20, 241071)
       (40, 248637)
       (60, 252515)
       (80, 259481)
       (100, 269438)
   };
   \addlegendentry{efg-paths}
   \addplot[color=purple, mark=diamond] coordinates {
       (20, 395125)
       (40, 404285)
       (60, 408571)
       (80, 417742)
       (100, 429485)
   };
   \addlegendentry{efg-total}
   \end{axis}
  \end{tikzpicture}
  \caption{Number of BWT runs with different indexes and varying number of covid19 strains. Here sequences refers to the bidirectional r-index, labels efg-nodes, efg-edges, efg-paths refer to the concatenation of node labels, paths of length 2 and paths of length 3, respectively. Label efg-total is the sum of the previous three numbers of runs. \label{fig:numberofruns}}
\end{figure}

The number of runs (the major factor affecting the space used by the indexes) for the bidirectional r-index of the collection and that of the concatenation of node labels are comparable. For edges and especially for paths of length 3 nodes the number of runs is significantly higher. Fortunately, the growth of these metrics when more strains are added is limited. This is not surprising, as the stains are highly similar and thus the added information content is limited and known to be correlated with the number of BWT runs~\cite{makinen2010storage}. 

Running times correlate with the index size comparison: with 100 strains, MEM finding using the the bidirectional r-index took 20 seconds, node MEM finding on the \textefg~took 31 seconds, edge MEM finding 97 seconds, and path MEM finding 217 seconds. The running times were measured on a server with Intel Xeon 2.9 Ghz processor and exclude index construction, which took 1 second for the patterns, 11 seconds for the bidirectional r-index of 100 strains, and 13 seconds for the additional indexes required by the corresponding \textefg. The construction of the \textefg~on 100 strains took 10 seconds. To speed-up MEM finding, we also tested switching the bidirectional r-index to a wavelet tree implementation of bidirectional BWT (\url{https://https://github.com/jnalanko/BD_BWT_index}). On the same 100 strains \textefg, the total time for MEM finding was 96 seconds, including index construction, which means 3.7 speed-up over the bidirectional r-index-based implementation.

\section{Discussion}\label{sect:discussion}

An alternative strategy to achieve the same goal as in our experiments is to encode the graph as an aggregate over the collection, apply MEM finding on the r-index, and report the distinct aggregate values on lexicographic MEM ranges to identify MEM locations in the graph \cite{GBPG22}. 
This approach is not comparable to ours directly, as the compressibility of the aggregates depends on the graph properties, and the indexable Elastic Founder Graph's size has not been analyzed with respect to $r$. Also, the two approaches use different MEM definitions. Our \Cref{def:mem} is symmetric and local, while the version used in earlier work with the $r$-index~\cite{Rosetal22,GBPG22} is asymmetric and semi-global: they define a MEM as a substring of a query that occurs in the text, but its query extensions do not appear in the text. For the purpose of chaining, only the symmetric definition yields connections to the Longest Common Subsequence problem~\cite{RCM23}.
For completeness, our implementation (\url{https://github.com/algbio/br-index-mems}) also supports this asymmetric MEM definition; our algorithms can be simplified for this case. 

We did not implement the general suffix tree-based approach to handle arbitrary long MEMs. From the experiments, we can see that already the case of paths of length 3 nodes handled by the generic MEM finding routine causes some scalability issues, and the suffix tree-based approach uses even more space. In our recent work \cite{RENM23}, we have solved pattern search in \textefgs\ using only edges, and our aim is to extend that approach to work with MEMs, so that the whole mechanism could work on top of a plain bidirectional r-index. Some of our results assume a constant-size alphabet. This assumption can be relaxed with additional data structures, but also a more careful amortized analysis may lead to better bounds.

\bibliography{biblio}
\end{document}